\documentclass[12pt,a4paper,final]{iopart}

\usepackage[normalem]{ulem}
\usepackage{exscale}
\usepackage{graphicx}
\usepackage{latexsym}
\usepackage{amsfonts}
\usepackage{amssymb}
\usepackage{times}
\usepackage[T1]{fontenc}
\usepackage{iopams}
\usepackage{amsthm}

\usepackage{epsfig,pstricks}

\newtheorem{theorem}{Theorem}
\newtheorem{lemma}{Lemma}
\newtheorem{corollary}{Corollary}

\begin{document}

\title{Local unitary symmetries and entanglement invariants}

\author{Markus Johansson} 
\address{ICFO-Institut de Ciencies Fotoniques, Mediterranean Technology Park, 08860 Castelldefels (Barcelona), Spain}
\address{Centre for Quantum Technologies, National University
of Singapore, 3 Science Drive 2, 117543 Singapore, Singapore }

\date{\today}
 
\begin{abstract}
We investigate the relation between local unitary symmetries and entanglement invariants of multi-qubit systems. The Hilbert space of such systems can be stratified in terms of states with different types of symmetry. We review the connection between this stratification and the ring of entanglement invariants and the corresponding geometric description in terms of algebraic varieties. On a stratum of a non-trivial symmetry group the invariants of the symmetry preserving operations gives a sufficient description of entanglement. Finding these invariants is often a simpler problem than finding the invariants of the local unitary group.
The conditions, as given by the Luna-Richardson theorem, for when the ring of such invariants is isomorphic to the ring of local unitary invariants on the stratum are discussed. 
As an example we consider symmetry groups that can be diagonalized by local unitary operations and for which the group action on each qubit is non-trivial.
On the stratum of such a symmetry the entanglement can be described in terms of a canonical form and the invariants of the symmetry preserving operations. This canonical form and the invariants are directly determined by the symmetry group.
Further, we briefly discuss how some recently proposed entanglement classification schemes capture symmetry properties.

\end{abstract}

\pacs{03.67.Mn, 03.65.Ud}

\maketitle
\section{Introduction}

Symmetries play an important role in many areas of physics. In quantum mechanics, entanglement makes possible local unitary symmetries of many particle systems that can not exist classically. These kinds of symmetries has been utilized for example in measurement based and topological quantum computation \cite{briegelraussdorff,hein,kitaev1,kitaev2,bombin1} and error correcting codes \cite{gottes}. 
In spin lattices such symmetries are relevant for phase transitions as well as string and topological orders \cite{bombin2,cirac1,cirac3,cirac4}.

The fact that many of the well studied entangled states including the Bell states, GHZ and W states have a non-trivial symmetry group has motivated the use of symmetries for characterization and classification of entanglement \cite{lindenpopescu, carteretlinden99, carteret}. For example, bipartite entanglement in relation to symmetries has been investigated in Refs. \cite{carteret, vollbrecht} and three-qubit entanglement and continuous symmetries has been studied in Ref. \cite{carteret}. 
More recently the relation between entanglement and symmetries for permutation invariant multi-qubit states has been explored \cite{hayashi,aulbach,markham,ribeiro,cenci1,cenci2,lyons,lewenstein}. The symmetry groups which are only exhibited by entangled systems can furthermore be used for entanglement verification protocols and witnesses \cite{geza,lyonswalck}. 

The classification and characterization of multipartite entanglement is a difficult problem due to the rich structure of different ways in which multiple subsystems can be entangled. Therefore, several approaches have been developed to give a meaningful structure to the set of entangled states. These include the use of entanglement invariants \cite{lindenpopescu,grassl,woot,carteretlinden99,linposud,woot2,sud,luq,luque05,OS05,OS06,djoko,OS09}, canonical forms \cite{acin,carteret00,acin2,dur,verstraete,kraus}, and geometric descriptions in terms of algebraic varieties \cite{acin,acin2,holweck}.
 
In this paper we investigate the relation between local unitary symmetry groups in multi-qubit systems and entanglement invariants. We review how such symmetries are related to different subsets of the entanglement invariants, and how this leads to a geometric description of the set of entangled states in terms of algebraic varieties. Furthermore, we use the theory of Luna stratifications \cite{luna73,alggeo} to see how the presence of a symmetry allows for a simplified description of entanglement in terms of the invariants of the group of symmetry preserving operations.
We apply the Luna-Richardson theorem \cite{lunarichardson79} to the context of multi-qubit systems and discuss the conditions for when symmetry determines the structure of the ring of polynomial entanglement invariants. In particular we study symmetry groups that are diagonalized by local unitary transformations and a description of entanglement in terms of canonical forms and invariants of symmetry preserving operations. These canonical forms and invariants are directly determined by the symmetry group and are closely related to the invariant comb approach for describing entanglement \cite{OS05,OS06,djoko,OS09}.
Using these results we then briefly discuss how different types of symmetries can be associated with different invariance properties and comment on the role of symmetries in some of the previously proposed classification schemes based on entanglement invariants \cite{viehmann,johanssonosterloh13,holweck}.

The paper is organized as follows.
We first briefly review some properties of entanglement invariants in Sect. \ref{sec1} followed by a review of how the Hilbert space of a multi-qubit system can be stratified into sets of states with different types of symmetries in Sect. \ref{sec2}. In Sect. \ref{sec3} we consider the groups of local unitary symmetry preserving operations and review the results from algebraic geometry that relates the invariants under action of this group to a description of the entanglement. 
In Sect. \ref{sec5} we treat the special case of non-trivial locally diagonalizable symmetries and describe the associated strata in terms of canonical forms and invariants under the symmetry preserving operations.  Finally, in Sect. \ref{sec6} the results are discussed and comments are made on the role of symmetries in some previously proposed classification schemes.

\section{Entanglement and symmetries}

\subsection{Entanglement invariants}
\label{sec1}
Any entanglement property of an $n$-qubit system is by definition invariant under some group $G$ of local operations that includes $\mathrm{SU(2)}^{\times{n}}$. Therefore, any function of such an entanglement property must be invariant under action of $G$ and thus constant on each $G$ orbit. To construct any possible function of entanglement we need a basic set of functions that can distinguish any two orbits, i.e., any difference in entanglement.

The orbits of a compact linear group $G$ acting on a real vector space are distinguished by the polynomial invariants \cite{popov}. In other words, for any two
$G$-orbits there is at least one invariant polynomial that takes different values on the orbits. However, this is not true for complex vector spaces.
Therefore, to get a set of invariants that distinguish the orbits of a complex vector space $X$, such as the $n$-qubit Hilbert space, one must represent the complex vectors as real vectors and consider the action of a real representation of $G$.
A polynomial invariant $I(|\psi\rangle)$ under such an action can be expressed as a function of the coefficients and the complex conjugated coefficients of the original complex vector $|\psi\rangle$ \cite{luque05}. A generating set of these invariants can be chosen from the bi-homogeneous polynomials. A bi-homogeneous polynomial $I(|\psi\rangle)$ scales under multiplication of $|\psi\rangle$ by a complex scalar $\lambda$ as $I(\lambda|\psi\rangle)=\lambda^{a}\lambda^{*b}I(|\psi\rangle)$ for some positive integers $a$ and $b$. The pair $(a,b)$ is the called the {\it bidegree} of $I(|\psi\rangle)$. 

Commonly one considers polynomial invariants of ${\mathrm{SU(2)}}^{\times{n}}$ and the subset of these polynomials invariant under ${\mathrm{SL(2)}}^{\times{n}}$.
The group ${\mathrm{SL(2)}}^{\times{n}}$ represents the stochastic local operations and classical communication (SLOCC), up to a global scaling factor of the state vector, and therefore the polynomials invariant under this group describe entanglement properties preserved by SLOCC. 
A bi-homogeneous polynomial invariant under ${\mathrm{SU(2)}}^{\times{n}}$ action is an invariant of ${\mathrm{SL(2)}}^{\times{n}}$ if and only if the bidegree $(a,b)$ satisfies $a=0$ or $b=0$, or if it is a product of such invariant polynomials. Since ${\mathrm{SL(2)}}^{\times{n}}$ is not a compact group its invariants do not distinguish all ${\mathrm{SL(2)}}^{\times{n}}$ orbits.
If at least one ${\mathrm{SL(2)}}^{\times{n}}$ invariant is non-vanishing for a state we say that it is SLOCC {\it semistable} \cite{mumford}. Otherwise it is called SLOCC {\it unstable}.




We denote the ring of polynomial invariants under action of $G$ on $X$ by $\mathbb{C}[X]^{G}$. All other functions of entanglement can in principle be constructed as functions on the ring of polynomial invariants.
An example of the more general functions that can be constructed is the field of ratios between polynomials in $\mathbb{C}[X]^{G}$. This so called field of fractions or {\it function field} of $\mathbb{C}[X]^{G}$ is denoted $\mathbb{C}(X)^{G}$ and includes also functions with singularities in lower dimensional subsets of the Hilbert space. Since the function field includes the ring of polynomial invariants it too gives a complete description of the entanglement.

For two qubits the ring of polynomial invariants is generated by the real valued $(1,1)$ polynomial that is the norm of the state vector and the complex valued $(2,0)$ polynomial that is the complex concurrence \cite{woot} together with its complex conjugate of bidegree $(0,2)$.
For greater number of qubits the number of generators increases and there are several with the same bidegree. Moreover, for a given generator it is often possible to add or subtract powers or products of other generators without changing the bidegree. Therefore, the set of bi-homogeneous generators is not given but a matter of choice.
Different choices gives different physical meanings to the generators and a natural way to choose is with respect to some entanglement property of interest. 

One case where such a choice must be made between capturing different properties of an entangled system is the selection of bidegree $(6,0)$ generator of the entanglement invariants of a four-qubit system.
For example, by adding or subtracting different fractions of the third power of the bidegree $(2,0)$ generator it can be chosen to distinguish between different degeneracy configurations \cite{lamata} of permutation invariant states \cite{viehmann}, or alternatively to distinguish between states exhibiting different topological phases \cite{johanssonosterloh13}.

For two or more qubits there is an uncountable number of $\mathrm{SU(2)}^{\times{n}}$ orbits. 
A simpler and more coarse grained classification of entanglement can be achieved by considering only the polynomial invariants and orbits under SLOCC. For two and three qubits the number of SLOCC orbits is finite. But for more than three qubits this gives again an uncountable number of orbits. 
This has lead to the development of different approaches to arrange the set of entangled states
into a finite number of classes.

One such approach is to consider the algebraic varieties associated with a set of invariants, i.e. the zero locus of each subset of the invariants \cite{acin,acin2,holweck}. Geometrically, an algebraic variety is similar to a manifold except that it may have singular points where the tangent space is not well defined. These varieties give a geometric structure to the set of entangled states. Since the number of generators is finite there always exist such geometric descriptions with a finite number of varieties.
However, different choices of the set of invariants leads to geometric descriptions with different physical meaning.

Another way to describe a set of entangled states is by canonical forms, i.e., a collection of basis vectors in terms of which any state in the set is expressible after a local unitary transformation \cite{acin,carteret00,acin2,kraus}, or SLOCC \cite{dur,verstraete}. 
Such a description is related to invariants if the canonical forms are chosen to describe the algebraic varieties associated with the invariants \cite{acin,acin2,holweck}.

\subsection{Symmetry groups and symmetry strata}
\label{sec2}

We say that an $n$-qubit system in state $|\psi\rangle$ has a local unitary symmetry if there exist a non-trivial group $H\subseteq\mathrm{U(1)}\times\mathrm{SU(2)}^{\times{n}}$ such that $g|\psi\rangle=|\psi\rangle$ for every $g\in{H}$. 
Let us establish some properties of the action of the local unitary group and the symmetry groups of states that will be needed in the Sect. \ref{sec3}.
The connected group $\mathrm{SU(2)}^{\times{n}}$ has no non-trivial connected Abelian normal subgroup and is therefore semisimple. A group like $\mathrm{U(1)}\times\mathrm{SU(2)}^{\times{n}}$ where the identity component is a product of a semisimple group and an algebraic torus is here called {\it reductive} to agree with the notation in Ref. \cite{lunarichardson79}. If $|\psi\rangle$ belongs to a closed orbit of the reductive group $G$ the symmetry group $H\subset{G}$ is reductive. Since $\mathrm{U(1)}\times\mathrm{SU(2)}^{\times{n}}$ is a compact group all orbits under its action are closed. In the following we let $G$ be the local unitary group. 

Next let us divide the symmetry groups into classes based on whether they can be transformed to each other by local unitary transformations.
Two symmetry groups $H$ and $H'$ are said to belong to the same symmetry class if they are conjugate, i.e., if there is a local unitary transformation $U$ such that $H'=UHU^{\dagger}$.
We denote the symmetry class that includes the group $H$ by $(H)$.
For a system of a finite number of qubits the different possible symmetries fall into a finite number of symmetry classes. States entangled in the same way now belong to the same symmetry class, but the converse is generally not true, states on different $G$ orbits can belong the same symmetry class.
In particular, the completely factorized states, i.e. the unentangled states all belong to the same symmetry class since they are on the same orbit. Hence, symmetries in any other class are exhibited only by entangled states.

The set of all states with symmetries belonging to the same symmetry class is called a symmetry stratum.
We denote such a stratum by $X_{(H)}$, where $H$ is the representative of the symmetry class.
The symmetry strata can be given a partial ordering based on inclusion in the closure of a bigger stratum. If $H_1\subseteq{H_2}$ the stratum $X_{({H_2})}$ exhibiting symmetry $H_2$ is included in the closure of the set $X_{({H_1})}$ exhibiting symmetry $H_1$. When discussing the closure of a stratum we refer to closedness in the Zariski topology. The Zariski topology is the one where closed sets are algebraic varieties, i.e., the zero locus of a family of polynomials.
We denote the closure of the symmetry stratum $X_{(H)}$ by $X^{(H)}$. It follows that the closure $X^{(H)}$ contains all states with a symmetry group that includes a subgroup conjugate to $H$.

The symmetry strata can range in size from a single local unitary orbit to an uncountable set of orbits. For any number of qubits there is a unique biggest stratum which is dense in state space called the principal stratum. For three or more qubits this stratum corresponds to the trivial symmetry group \cite{carteretlinden99}. 
Since the symmetries of completely factorized states belong to the same class all such states belong to the same stratum. Hence, all other strata contain only entangled states. The symmetry stratification for the special case of permutation invariant states has been studied in \cite{cenci1,cenci2}.

We can give a $G$-invariant description of the symmetry strata. The group quotient $X/G$ maps each symmetry stratum to a corresponding stratum in the orbit space. This stratification of the orbit space is called a Luna stratification \cite{luna73,alggeo} and each Luna stratum  $Z_{(H)}\equiv{X}_{(H)}/G$ is the set of closed $G$-orbits of states whose symmetry groups are conjugate to $H$.

The Luna strata can be given the same partial ordering as the symmetry strata. 
A Luna stratum $Z_{(H_{1})}$ is contained in the closure of a Luna stratum $Z_{(H_{2})}$ if and only if $H_{2}\subset{H_{1}}$. We denote the closure of $Z_{(H)}$ by $Z^{(H)}$. The largest stratum corresponding to the trivial symmetry group is again called he principal stratum.
It follows from the previous discussion that $Z^{(H)}$ is the set of closed $G$-orbits of states with a symmetry group that includes a subgroup conjugate to $H$.
The set $Z_{(H)}$ is a dense open subset of $Z^{(H)}$ \cite{lunarichardson79}.

The closure of each Luna stratum in the orbit space is an algebraic variety just like the closure of a symmetry stratum in the Hilbert space. Since we are in the orbit space, the closure of each stratum correspond to the vanishing of a subring of polynomials in $\mathbb{C}[X]^{G}$.
We are interested in the case where $G$ acts on $X$ such that $X$ is a rational finite-dimensional $G$-module since this is true for the action of the local unitary group on the $n$-qubit Hilbert space. In this case the closure $Z^{(H)}$ of each Luna stratum is an irreducible variety, i.e., it is not the union of two smaller varieties \cite{schwarz}. 
The closure of the principal stratum is the full Hilbert space, and thus no polynomials are zero everywhere on this stratum. If $Z^{(H_{2})}$ is in the closure of $Z^{(H_{1})}$, i.e., if $H_{1}\subset{H_{2}}$, the subring of vanishing polynomials on $Z^{(H_{1})}$ is contained in the subring of vanishing polynomials on $Z^{(H_{2})}$.

The non-vanishing polynomials on the closure $Z^{(H)}$ of a Luna stratum is the so called coordinate ring $\mathbb{C}[Z^{(H)}]$ of $Z^{(H)}$.  
A coordinate ring $\mathbb{C}[Z^{(H_{2})}]$ is contained in $\mathbb{C}[Z^{(H_{1})}]$ if $H_{1}\subset{H_{2}}$.
Note that this means that the set of orbits with symmetry class $(H_2)$ is described by a smaller number of parameters than the orbits with symmetry class $(H_1)$.
The most symmetric states, i.e., those with symmetry groups not contained in any other symmetry group, thus belong to the smallest strata of Hilbert space with relatively specific entanglement properties. Examples of this are the highly symmetric two-qubit Bell state and three-qubit GHZ state which both belong to symmetry strata containing a single local unitary orbit. 

To each symmetry group we can thus associate a subring of entanglement invariants, the coordinate ring of the closure of the symmetry stratum.
Furthermore, when there is a choice in selecting a set of bi-homogeneous polynomials from $\mathbb{C}[X]^{G}$ to describe the entanglement of the system one may choose them to reflect the symmetry stratification by selecting them from the coordinate rings of the different strata.

\subsection{Symmetry preserving operations and their invariant rings}
\label{sec3}
Given a symmetry group $H$ of an $n$-qubit system in state $\mid\psi\rangle$, we consider the subgroup of $G$ that preserves the symmetry of the system while not necessarily preserving the state. This subgroup is the normalizer group $N_{G}(H)$ of $H$ in ${G}$ which is defined as the group of all $g\in{G}$ such that $gH=Hg$, i.e., $g$ that commutes with the group $H$ as a whole. If $H$ is a closed reductive subgroup of the reductive group $G$, $N_{G}(H)$ is a reductive group \cite{lunarichardson79}.

In each symmetry stratum $X_{(H)}$ we can select the subset $X_{H}\subset{X_{(H)}}$ of states for which $H$ is the symmetry group. 
Every state in $X_{(H)}$ can by definition be brought to $X_{H}$ by a local unitary transformation. 
If two elements of $X_{H}$ are in the same closed $G$-orbit they can be transformed to each other by elements of the normalizer group $N_{G}(H)$ \cite{lunarichardson79}. This means that every $G$ orbit in $X_{(H)}$ contains a unique $N_{G}(H)$ orbit in $X_H$. 
Next let $X^{H}$ denote the set of states such that $H$ is a subgroup of the symmetry group. From the discussion in section \ref{sec2} we see that $X^{H}$ is the closure of $X_{H}$. Every closed $G$ orbit in $X^{(H)}$ meets $X^{H}$, but they do not necessarily meet in a unique $N_G(H)$ orbit. There could be more than one $N_G(H)$ orbit in a $G$ orbit of $X^{(H)}$.

To the symmetry group $H$ we can associate the ring $\mathbb{C}[X^{H}]^{N_{G}(H)}$ of polynomial functions on $X^{H}$ that 
are invariant under action of $N_{G}(H)$. In other words, these are the polynomials which are invariant under any symmetry preserving operation in $G$.
We introduce a $N_{G}(H)$ invariant description of $X_{H}$ by defining $Z_{H}\equiv{X}_{H}/N_{G}(H)$ as the set of closed $N_{G}(H)$ orbits whose symmetry group is $H$. Let $Z^{H}$ be the closure of $Z_{H}$. 
The coordinate ring on $Z^{H}$ is the ring $\mathbb{C}[Z^{H}]\equiv\mathbb{C}[X^{H}]^{N_{G}(H)}$ of $N_{G}(H)$-invariant polynomials on $X^{H}$. Since $N_{G}(H)$ is a compact group all orbits are closed and any two orbits are distinguished by $\mathbb{C}[Z^{H}]$.

Since every $G$ orbit in $X_{(H)}$ contains a unique $N_{G}(H)$ orbit, there is a morphism from $Z^H$ to
$Z^{(H)}$ which is an isomorphism from $Z_H$ to $Z_{(H)}$ but not necessarily from all of $Z^H$ to
$Z^{(H)}$. The reason why the isomorphism may not be valid on all of $Z^H$ is that a $G$ orbit in $X^{(H)}$ may contain more than one $N_{G}(H)$ orbit. A morphism which defines an isomorphism outside a lower dimensional subset is called a {\it birational} morphism and $Z^H$ is said to be {\it birationally equivalent} to
$Z^{(H)}$. Further, this birational morphism means that the coordinate ring
$\mathbb{C}[Z^H]$ which distinguishes $N_G(H)$ orbits on $X^{H}$ also distinguishes inequivalently entangled states on $X_H$ but not necessarily on all of $X^H$. Thus, $\mathbb{C}[Z^H]$ is sufficient to describe the entanglement properties of $X_H$. 
Since every state in $X_{(H)}$ is $G$ equivalent to a state in $X_H$
a description of the entanglement properties of $X_{H}$ gives a description of entanglement in $X_{(H)}$.
For a principal stratum with trivial symmetry group, the normalizer is the full group $G$. However, for a non-trivial symmetry group the normalizer $N_G(H)$ is a subgroup of $G$. 
In this case, instead of finding the ring of $G$ invariants, which can be a difficult problem, one can solve the often easier problem of finding the $N_G(H)$ invariants.

The ring $\mathbb{C}[Z^H]$ includes a subring isomorphic to $\mathbb{C}[Z^{(H)}]$, but is not isomorphic to $\mathbb{C}[Z^{(H)}]$ unless $Z^H$ is isomorphic to
$Z^{(H)}$. The existence of a birational map between irreducible varieties nevertheless implies an isomorphism between their function fields. The function field of the coordinate ring $\mathbb{C}[Z^{(H)}]$ on a stratum closure is defined as the ring of quotients of elements in $\mathbb{C}[Z^{(H)}]$ and is denoted $\mathbb{C}(Z^{(H)})$.

\begin{lemma}
Two irreducible varieties $X$ and $Y$ are birationally equivalent if and only if there is an isomorphism of the function fields $\mathbb{C}(X)\cong{\mathbb{C}(Y)}$ which is the identity on the field $\mathbb{C}$.
\end{lemma}
\begin{proof}
See Ref. \cite{cox}
\end{proof}

One can make a stronger claim if there is a birational morphism $\varphi:X\to{Y}$ and $Y$ is {\it normal}. A irreducible variety $Y$ is normal if its coordinate ring $\mathbb{C}[Y]$ is integrally closed, i.e. if any ring in the function field $\mathbb{C}(Y)$ which includes $\mathbb{C}[Y]$ and is finitely generated as a module over $\mathbb{C}[Y]$ is $\mathbb{C}[Y]$ itself.
In this case an isomorphism exists.

\begin{lemma}
\label{lemmq}
Let $\varphi:X\to{Y}$ be a surjective birational morphism of irreducible algebraic varieties. If $Y$ is a normal variety, $\varphi$ is an isomorphism of varieties.
\end{lemma}
\begin{proof}
See. Ref. \cite{luna75}
\end{proof}
What lemma \ref{lemmq} implies is that if $Z^{(H)}$ is a normal variety and if $Z^H$ is irreducible, there is an isomorphism between $Z^H$ and $Z^{(H)}$ and an isomorphism the coordinate rings $\mathbb{C}[Z^{H}]$ and $\mathbb{C}[Z^{(H)}]$. This is the essence of the Luna-Richardson theorem \cite{lunarichardson79}. Adapted to our context the theorem says the following

\begin{theorem}
\label{luna}
Let $G$ be a reductive group and $H\subset G$ a symmetry group. Assume that $Z^{(H)}$ is a normal irreducible variety and that $Z^H$ is irreducible. Denote $\pi_{X^{(H)}}:X^{(H)}\to{Z^{(H)}}$ and $\pi_{X^H}:X^H\to{Z^H}$. Let 
$\chi:Z^H\to Z^{(H)}$ be the unique morphism such that such that $\chi(\pi_{X^H}(x))=\pi_{X^{(H)}}(x)$ for every $x\in X^{(H)}$. Then $\chi$ is an isomorphism of varieties and $\mathbb{C}[Z^{H}]$ is isomorphic to $\mathbb{C}[Z^{(H)}]$.
\end{theorem}
\begin{proof}
See Ref. \cite{lunarichardson79}.
\end{proof}
The Luna-Richardson theorem tells us that if $\mathbb{C}[Z^{(H)}]$ is integrally closed it is isomorphic to $\mathbb{C}[Z^{H}]$. In this case $\mathbb{C}[Z^{H}]$ gives a description of the entanglement in all of $X^{H}$. 

The above discussion including theorem \ref{luna} was made for $G=\mathrm{U(1)}\times\mathrm{SU(2)}^{\times{n}}$, but the main points only require $G$ to be compact and reductive and thus are true for any compact reductive subgroup of the local unitary group containing $\mathrm{SU(2)}^{\times{n}}$. In section \ref{sec5b} we consider such subgroups to clarify the relation between certain symmetries and $\mathrm{SU(2)}^{\times{n}}$ invariants. 

While we have seen that there exists descriptions of the entanglement in terms of $N_G(H)$ invariants on $X_H$ and that there under some circumstances is an isomorphism between the ring of $N_G(H)$ invariants and the ring of $G$ invariants on the closure of the symmetry stratum, we have no general method for constructing theses invariants and isomorphisms.

\section{Symmetry groups diagonalizable by local unitary operations}
\label{sec5}
\subsection{Local unitary diagonalizability and canonical forms}
\label{seco}
Finding all the possible symmetries of an $n$-qubit system is in general a difficult task. However, a subset of the symmetries can still give a useful partial description.
Here we consider the case where the symmetry group $H$ can be diagonalized by local unitary operations. In the symmetry strata of such a group $H$ the set $Z^{H}$ can be easily described in terms of canonical forms. 

In particular we consider the subset of these groups that act non-trivially on all the qubits, or more precisely, symmetry groups contained in $\mathrm{U(1)}\times\mathrm{SU(2)}^{\times{n}}$ where all elements can be diagonalized by the same local unitary operations and where for each qubit subspace at least one element of the group acts differently from $\pm{1}$. 
In this case the invariants of the symmetry preserving operations $\mathbb{C}[Z^{H}]$ can be constructed from the canonical form in a relatively easy way. Thus, for locally diagonalizable symmetry groups we can directly construct the description of the entanglement outlined in Sect. \ref{sec3}.

This subset of symmetries still capture many physically interesting cases.
An example is the non-trivial symmetry groups of permutation invariant states. For these states an element of a symmetry group occurs only if all other elements related to it by permutations of the single qubit actions are also in the symmetry group.
Therefore, in this case any non-trivial symmetry group has a non-trivial action on every qubit.

Given a locally diagonalizable group $H$ our first step is to find the set of states $X^{H}$, i.e., the fix points of $H$. Therefore, we study the equations $g|\psi\rangle=|\psi\rangle$ where $g$ is an element of $H$. Since $H$ is locally diagonalizable we can choose a basis where $H$ is diagonal. A group element in such a basis is of the form $g=e^{i\phi_{1}\sigma_{z}}\times{e^{i\phi_{2}\sigma_{z}}}\times\dots\times{e^{i\phi_{n}\sigma_{z}}}\times{e^{i\theta}}$
, where $\phi_{k},\theta\in{\mathbb{R}}$, and can thus be fully described by the set $\{\phi_{k},\theta\}$.
Let us denote the basis vectors $|s_{j1}\dots{s_{jn}\rangle}$ where $s_{jk}$ is either $0$ or $1$.
The expansion of a state $|\psi\rangle$ satisfying $g|\psi\rangle=|\psi\rangle$ in these basis vectors gives

\begin{equation}|\psi\rangle=\sum_{j=1}^{L}c_{j}|s_{j1}\dots{s_{jn}\rangle},\end{equation}
for some positive integer $L$.
Each basis vector that is included in the expansion of $|\psi\rangle$ is individually preserved by the group action. Therefore, to each locally diagonalizable symmetry group $H$ we can associate the set $S(H)$ of basis vectors that are preserved by action of the group. 

For each basis vector in $S(H)$ action by a group element gives the condition $e^{i(\sum_{k=1}^{n}\phi_{k}(-1)^{s_{jk}}+\theta)}=1$.
This is equivalent to an equation in the exponents $\sum_{k=1}^{n}\phi_{k}(-1)^{s_{jk}}+\theta=2a_{j}\pi$ where $a_{j}$ is an integer.
A state in $X^{H}$ with $L$ terms in the basis expansion gives us a system of $L$ equations.

\begin{eqnarray}\label{yutt}\left(\begin{array}{ccccc}
(-1)^{s_{11}} & \cdots & (-1)^{s_{1n}} & 1\\
(-1)^{s_{21}} & \cdots & (-1)^{s_{2n}} & 1\\
\vdots & \ddots &\vdots & \vdots\\
(-1)^{s_{L1}} & \cdots & (-1)^{s_{Ln}} & 1\\
\end{array}\right)\left(\begin{array}{c}
\phi_{1}  \\
\vdots\\
\phi_{n}\\
\theta\\
\end{array}\right)\!=\!\left(\begin{array}{c}
2\pi{a}_{1}\\
2\pi{a}_{2}\\
\vdots\\
2\pi{a}_{L}\\
\end{array}\right).\nonumber\\
\end{eqnarray}
If the group $H$ is discrete and $|\psi\rangle\in{X_{H}}$ Eq. \ref{yutt} has a unique solution $\{\phi_{k},\theta\}$ for a given choice of integers $a_{j}$. In this case there is $n+1$ linearly independent rows on the left hand side. While $S(H)$ for a discrete group $H$ contains at least $n+1$ elements it may contain more. There are thus in general several sets of $n+1$ linearly independent rows which each uniquely define the group $H$ through the $n+1$ linearly independent solutions of Eq. \ref{yutt} for different choices of $a_{j}$. 
A combinatorial algorithm to find such linearly independent sets corresponding to this kind of discrete symmetry groups was described in \cite{johansson12}.

If there are less then $n+1$ linearly independent rows in the left hand side of Eq. \ref{yutt} the $\{\phi_{k},\theta\}$ are dependent variables and there is a continuous set of solutions. Thus, in this case $H$ is a continuous group. As for the case of discrete groups there may be several sets of linearly independent rows which each uniquely define the same group $H$.
An important distinction to be made is whether $\theta$ is a discrete or continuous variable. If $\theta$ is discrete $H\subset{D}\times{\mathrm{SU(2)}}^{\times{n}}$, where $D$ is some discrete cyclic subgroup of ${\mathrm{U(1)}}$. If $\theta$ is continuous $H$ is not contained in such a group and but only in ${\mathrm{U(1)}}\times{\mathrm{SU(2)}}^{\times{n}}$.
Moreover, if $\theta$ is continuous all invariants in $\mathbb{C}[X^{(H)}]^{G}$ have bidegrees of the type $(a,a)$ for $a\in{\mathbb{N}}$ \cite{johansson12,johanssonosterloh13}.

By definition every state in $X_{H}$ is in $span(S(H))$ and every state for which $H$ is a subgroup of the symmetry group is also in $span(S(H))$. Therefore, if $X_{H}$ is non-empty and dense in $span(S(H))$ it follows that $span(S(H))=X^{H}$. So far we have considered locally diagonalizable symmetry groups in generality but we now make a restriction to the groups with non-trivial action on all qubits.

\begin{lemma}
Consider a symmetry group $H$ which is diagonalized by local unitary operations.
If two basis vectors in the set $S(H)$ associated with $H$ differ from each other in only the $k$th entry corresponding to the $k$th qubit it follows that each element of the symmetry group $H$ must act trivially on the $k$th qubit. 
\end{lemma}
\begin{proof}
Consider the equation $e^{i(\sum_{k=1}^{n}\phi_{k}(-1)^{s_{jk}}+\theta)}=1$ for the phase factor resulting from action of $h\in{H}$ on a vector in $S(H)$. If two vectors differ only in the $k$th entry the ratio of their phase factors is $e^{i2\phi_{k}}$. Since both vectors are in $S(H)$ it follows that $e^{i2\phi_{k}}=1$ and thus $e^{i\phi_{k}}=\pm{1}$. 
\end{proof}
For a group that acts non-trivially on every qubit if follows that $X_{H}$ is non-empty and dense in $span(S(H))$.

\begin{theorem}
\label{kullo}
Assume that $H$ is the largest group which can be diagonalized by local unitary operations and is such that $H|\psi\rangle=|\psi\rangle$ for every $|\psi\rangle\in{span(S(H))}$. Assume that the action of $H$ is non-trivial on every qubit.  Then for almost all states in $span(S(H))$, $H$ is the full symmetry group. 

\end{theorem}
\begin{proof}
Assume that $|\psi\rangle$ is a state for which the group $H$ satisfies $h|\psi\rangle=|\psi\rangle$ for every $h\in{H}$. 
Then let $U\notin{H}$ be such that $U|\psi\rangle=|\psi\rangle$. We consider $U$ in the basis for which $H$ is diagonal and write it on a form where its action on the first qubit is distinguished from its action $\tilde{U}$ on the remaining qubits

\begin{equation}U= \left( \begin{array}{cc}
\!\!\alpha & \!\beta \\
\!\!-\beta^* &\! \alpha^*  \end{array} \!\right)\times{\tilde{U}}.\end{equation}
The state vector $|\psi\rangle$ in the same basis is expanded in terms of the basis vectors $S(H)$. The expansion can be divided into the collection of terms $|0\rangle\otimes|\theta\rangle$ for which the state of the first qubit is $|0\rangle$ and the collection of terms $|1\rangle\otimes|\varphi\rangle$ with the state of the first qubit is $|1\rangle$. 
\begin{equation}|\psi\rangle=|0\rangle\otimes|\theta\rangle+|1\rangle\otimes|\varphi\rangle
\end{equation}
Here $\langle\theta|\theta\rangle$ and $\langle\varphi|\varphi\rangle$ are assumed to be non-zero. If $\langle\varphi|\varphi\rangle=0$ or $\langle\theta|\theta\rangle=0$ elements of the symmetry group cannot have a non-diagonal component acting on the first qubit.

Since $U$ must satisfy $U|\psi\rangle=|\psi\rangle$ it follows specifically that

\begin{equation}|0\rangle\otimes|\theta\rangle=\alpha|0\rangle\otimes\tilde{U}|\theta\rangle+\beta|0\rangle\otimes\tilde{U}|\varphi\rangle
\end{equation}
The assumption of non-trivial action implies that $\langle\theta|\varphi\rangle=0$.
Taking the norm of both sides therefore gives $|\beta|^2\langle\theta|\theta\rangle=|\beta|^2\langle\varphi|\varphi\rangle$. This is possible for non-zero $\beta$ only if $\langle\theta|\theta\rangle=\langle\varphi|\varphi\rangle$. This condition on the form of $|\psi\rangle$ is satisfied only by a subset of the states in $span(S(H))$ of measure zero. 
Hence, for almost all states in $span(S(H))$, $H$ is the full symmetry group.

\end{proof}

The proof of Theorem \ref{kullo} implies the following corollary.

\begin{corollary}
Let a state $|\psi\rangle$ have a symmetry group $H$ that is not diagonalizable by local unitary operations but contains a subgroup $H'$ that is diagonalizable by local unitary transformations and which has a non-trivial action on every qubit. Then $|\psi\rangle$ has at least one maximally mixed one-qubit reduced density matrix.
\end{corollary}

\begin{proof}
The necessary requirement for a symmetry group which is not locally diagonalizable for the $k$th qubit subspace within the stratum of locally diagonalizable symmetry group with non-trivial action on every qubit in Theorem \ref{kullo} is equivalent to the reduced density matrix of the $k$th qubit being maximally mixed.
\end{proof}

From Theorem \ref{kullo} we see that if $H$ acts non-trivially on each qubit it follows that $X^{H}=span(S(H))$ and $X_{H}$ is dense in $span(S(H))$.
Since each state in $X^{(H)}$ can be transformed by local unitary operations to $X^{H}$ it follows that $S(H)$ corresponds to a canonical form for all of $X^{(H)}$.
However, as described above there are in general proper subsets $s_{i}(H)\subset{S(H)}$ such that $H$ is the maximal locally diagonalizable subgroup of the symmetry group of each state expanded in the vectors of $s_{i}(H)$. In this case there are subsets of $X^{(H)}$ for which $s_{i}(H)$ corresponds to a canonical form.
We define the following set of states
\begin{eqnarray}\mathcal{S}_{i}(H)\equiv\left\{|\psi\rangle=\sum_{|j\rangle\in{s_{i}(H)}}c_{j}|j\rangle \phantom{u}| \phantom{u} c_{j}\neq{0}\phantom{u} \forall{j}\right\}.\end{eqnarray}

If $s_{i}(H)$ is such that for the states of $\mathcal{S}_{i}(H)$ the rows in the left hand side of Eq. \ref{yutt} are linearly independent and $\theta$ is uniquely defined for any choice of right hand side but no subset of $s_{i}(H)$ has the same property, the states of $\mathcal{S}_{i}(H)$ are {\it irreducibly balanced} in the sense of Ref. \cite{OS09}. The entanglement of such states has been extensively studied in Refs. \cite{OS05,OS06,djoko,OS09,johanssonosterloh13}. 
If $H$ is a discrete locally diagonalizable group, then each $s_{i}(H)$ has at least $n+1$ elements.

While the sets $s_{i}(H)$ correspond to canonical forms on $X^{(H)}$ they also correspond to weight polytopes of the maximal Abelian subgroup of ${\mathrm{SU(2)}}^{\times{n}}$ which is diagonal in the same basis as $H$.
To see this, consider the $j$th row of the left hand side of Eq. \ref{yutt} and exclude the entry in the last column.
This part of the row can be viewed as a vector $v_{j}$ in $\mathbb{R}^{n}$. In this way each basis vector in $s_{i}(H)$ corresponds to such a vector in $\mathbb{R}^{n}$.
As described in \cite{johansson12} the vectors $v_{j}$ are the weight vectors of the maximal Abelian subgroup of ${\mathrm{SU(2)}}^{\times{n}}$ which is diagonal in the chosen basis. In other words, the entries of each vector are the eigenvalues of the action of a set of generators of the group. As such they describe the infinitesimal action of the group. The weight polytope is the polytope in $\mathbb{R}^{n}$ spanned by the collection of these vectors.

The solutions to Eq. \ref{yutt} can be found through a Gaussian elimination of the left hand side to row echelon form. If $H\subset{D}\times{\mathrm{SU(2)}}^{\times{n}}$, where $D$ is some discrete cyclic subgroup of ${\mathrm{U(1)}}$, i.e., if $\theta$ takes only a discrete set of values, there are integers $z_{j}\in{\mathbb{Z}}$ corresponding to this Gaussian elimination such that 

\begin{equation}\label{gnuu}\sum_{j}z_{j}v_{j}=0.\end{equation}
For the local SLOCC semistable irreducibly balanced states, all the $z_{j}$ can be chosen positive \cite{mumford,OS09}.  
This means that the convex hull of the weight polytope contains the origin in $\mathbb{R}^{n}$.
The irreducibly balanced states for which all $z_j$ are positive correspond to polytopes where none of the vectors can be removed without reducing the polytope to one which does not contain the origin in its convex hull.
Thus, these irreducibly balanced states are both minimal non-redundant sets of basis vectors which are SLOCC semistable \cite{OS09} and minimal non-redundant sets of vectors which define locally diagonalizable subgroups of $D\times{\mathrm {SU(2)}}^{\times{n}}$ where $|D|=\sum{z_j}$. The irreducible balancedness implies that these groups are not subgroups of any locally diagonalizable group contained in $D\times{\mathrm {SU(2)}}^{\times{n}}$ where $|D|$ is finite.
If a set of vectors $S(H)$ which is not irreducibly balanced but contain irreducibly balanced sets defines a locally diagonalizable group $H$, this group is a subgroup
of the locally diagonalizable groups defined by the irreducibly balanced sets contained in $S(H)$. 
The connection between irreducibly balanced states and SLOCC invariants has been described in \cite{OS09}, as a part of the invariant-comb approach for constructing polynomial entanglement invariants \cite{OS05,OS06,OS09,djoko}.
The relation between such irreducibly balanced states and locally diagonalizable symmetry groups has been used in the study of topological phases \cite{johansson12,johanssonosterloh13}.

If for an irreducibly balanced state $z_{j}$ must be chosen from both positive and negative integers, the state is SLOCC unstable and the weight polytope does not contain the origin in its convex hull. It is however contained in the affine hull \cite{johansson12,johanssonosterloh13}. Such irreducibly balanced states have been termed affinely balanced in \cite{johansson12}. These also define locally diagonalizable groups in $D\times{\mathrm {SU(2)}}^{\times{n}}$ where $|D|=\sum{z_j}$ which are not subgroups of any other locally diagonalizable group in $D\times{\mathrm {SU(2)}}^{\times{n}}$ where $|D|$ is finite.

\subsection{Invariants under symmetry preserving operations}
\label{sec5b}
We have seen how a locally diagonalizable symmetry group $H$ with non-trivial action on all qubits determines a canonical form on its symmetry stratum and that this canonical form gives direct information about the behaviour of the state under SLOCC operations. This canonical form in turn determines the algebra of entanglement invariants on $X_H$.
To see how we must consider the group of symmetry preserving operations $N_{G}(H)$ as described in Sect. \ref{sec3}.

\begin{theorem}
\label{ullko}
The normalizer $N_{G}(H)$ of a locally diagonalizable symmetry group $H$ which does not act trivially on any qubit is the maximal Abelian subgroup of $G$, which contains $H$ and is diagonalizable by local unitary operations, in product with a finite group of spin flips.
\end{theorem}
\begin{proof}

Let $f=e^{-i\varphi}\times{}f_{1}\times{f_{2}}\times\dots\times{f_{n}}\in{N_{G}(H)}$, where $f_i\in{\mathrm{SU(2)}}$.
By definition $fH=Hf$, i.e., for every element\\ $h=e^{-i\theta}\times{}h_{1}\times{h_{2}}\times\dots\times{h_{n}}\in{H}$ there is an element\\ $h'=e^{-i\theta'}\times{}{h'}_{1}\times{{h'}_{2}}\times\dots\times{{h'}_{n}}\in{H}$
such that $f_{k}h_{k}={h'_{k}f_{k}}$ up to a factor $\pm 1$ for every $k$. The factor $\pm 1$ can be absorbed in $\theta'$.
Let $h$ be an element with non-trivial action on the $k$th qubit and let\\ $f_{k}\!=\!\left( \begin{array}{cc}
\!\!\alpha & \!\beta \\
\!\!-\beta^* & \!\alpha^*  \end{array} \!\right)$, $h_{k}\!=\!\left( \begin{array}{cc}
\!e^{i\phi_{k}} & \!0 \\
\!0 & \!e^{-i\phi_{k}}  \end{array} \!\!\right)$, $h'_{k}\!=\!\left( \begin{array}{cc}
\!e^{i\phi_{k}'} & \!0 \\
\!0 & \!e^{-i\phi_{k}'}  \end{array} \!\!\right)$.\\
Then the condition $f_{k}h_{k}=h'_{k}f_{k}$ is equivalent to 
\begin{equation}\left( \begin{array}{cc}
\!\!\alpha & \!\beta \\
\!\!-\beta^* & \!\alpha^*  \end{array} \!\right){\!\!\!}\left( \begin{array}{cc}
\! e^{i\phi_{k}} & \! 0 \\
\! 0 & \! e^{-i\phi_{k}}  \end{array} \!\!\right){\!\!}={\!}\!\left( \begin{array}{cc}
\!e^{i\phi_{k}'} & \!0 \\
\!0 & \!e^{-i\phi_{k}'}  \end{array} \!\!\right){\!\!\!}\left( \begin{array}{cc}
\!\!\alpha & \!\beta \\
\!\!-\beta^* & \!\alpha^*  \end{array} \!\right).\end{equation}
This implies that
$\alpha e^{i\phi_{k}}=\alpha e^{i{\phi}_{k}'}$ and $\beta e^{-i\phi_{k}}=\beta e^{i{\phi}_{k}'}$ must be satisfied which is possible for non-zero $\alpha$ and non-zero $\beta$
only if $\phi_{k}=0$ or $\pm\pi$. However this contradicts the assumption that $h$ acts non-trivially on the $k$th qubit.
Hence, either $\alpha$ or $\beta$ must be zero. Any matrix for which $\beta=0$ commutes with $H$.

If $\alpha=0$ it follows that $e^{-i\phi_{k}}=e^{i{\phi}_{k}'}$, i.e., $h_{k}={h'}_{k}^*$. 
By assumption $\phi_{k}\neq m\pi$ for integer $m$. Thus, $\alpha=0$ implies that there is a subset $Z$ of the qubits such that for every $h\in{H}$ there is an $h'\in{H}$ where $h_{k}={h'}_{k}^*$ for each $k\in{Z}$ and $\sum_{k\in{Z}}(-1)^{s_{jk}}\phi_k=m\pi$ for each $j$. Since $h_{k}={h'}_{k}^*=\sigma_{x}{h'}_{k}\sigma_{x}$, this implies that for every vector $v$ in $S(H)$ there is a another vector in $S(H)$ related to $v$ by action of $\sigma_{x}$ on each of the qubits in the subset $Z$.
Then the operations of simultaneous $\sigma_{x}$ on any such subset of qubits is included in the normalizer.
Thus, $N_{G}(H)$ is the product of the maximal Abelian group containing $H$ that is diagonalizable by local unitary operations and a finite group of spin flips.

\end{proof}
Theorem \ref{ullko} establishes the general form of $N_{G}(H)$ for locally diagonalizable $H$ with non-trivial action on every qubit.
The polynomials in $\mathbb{C}[X^H]^{N_{G}(H)}$
 are at minimum invariant under the action of the maximal Abelian subgroup of $G$ containing $H$. Therefore, any such polynomial is an algebraic combination of monomials invariant under the maximal Abelian subgroup of $G$. 
We can now consider the rings $\mathbb{C}[X^H]^{N_G(H)}$ for different choices of $G$.

If $G={D\times\mathrm{SU(2)}}^{\times{n}}$, where $D$ is a cyclic subgroup of ${\mathrm{U(1)}}$ of order $|D|$, the symmetry groups included in $G$ are those with elements, given as $\{\phi_{k},\theta\}$, for which $\theta$ are multiples of $\frac{2\pi}{|D|}$.
As described in Sect. \ref{seco} the symmetry groups of this kind are those for which there are $s_{i}(H)\subset{S(H)}$ corresponding to a $\mathcal{S}_{i}(H)$ of irreducibly balanced states. For each such $s_i(H)$ there is a unique lowest degree monomial, up to complex conjugation, of the form

\begin{eqnarray}m_{i}\equiv{\prod_{j=1}^{L}c_{j}^{\frac{1}{2}(|z_{j}|+z_{j})}c_{j}^{*\frac{1}{2}(|z_{j}|-z_{j})}},\end{eqnarray}
where $z_{j}$ is the integer multiplying the vector $v_{j}$ in Eq. \ref{gnuu} for the set of vectors $v_{j}$ in $\mathbb{R}^{n}$ associated with $s_{i}(H)$, and $c_{j}$ is the coefficient of the corresponding basis vector in Hilbert space \cite{OS09,johanssonosterloh13}. 
Such an $m_i$ is invariant under the maximal locally diagonalizable Abelian subgroup containing $H$ of ${D\times\mathrm{SU(2)}}^{\times{n}}$, with $|D|=\sum{z_j}$. This follows since Eq. \ref{gnuu} is a condition for invariance of a monomial under the infinitesimal action of this maximal Abelian subgroup. More precisely, the conditions for invariance under action of the generators of the maximal Abelian subgroup is a set of linear equations which can be expressed as a sum of weight vectors $\sum_{j=1}^Lz_jv_j=0$ where each coefficient $c_j$ or $c_j^*$ in $m_i$ correspond to a term in the sum. The complex conjugated coefficients in $m_i$ correspond to the negative $z_j$ since these are the coefficients of the state related by a universal spin flip operation to the original state \cite{buzek99,johanssonosterloh13}.
We call this kind of monomials {\it irreducibly balanced}. The ring of invariants of the maximal Abelian subgroup of ${D\times\mathrm{SU(2)}}^{\times{n}}$ is generated by the irreducibly balanced monomials together with the monomials $|c_{j}|^{2}$.

The relation between the irreducibly balanced sets of vectors and locally diagonalizable subgroups of ${D\times\mathrm{SU(2)}}^{\times{n}}$ for discrete cyclic $D$ described in Sect. \ref{seco} thus carries over to the irreducibly balanced monomials. Each irreducibly balanced monomial $m_i$ uniquely defines $H$ through $s_i(H)$. Conversely, each locally diagonalizable symmetry group $H$ for which $\theta$ takes discrete values uniquely defines a set of irreducibly balanced monomials.

If $N_{G}(H)$ is the maximal Abelian subgroup containing $H$ the ring $\mathbb{C}[Z^H]$ is generated by the irreducibly balanced monomials and the monomials $|c_{j}|^{2}$. 
If on the other hand $N_{G}(H)$ includes a group $G^{z}$ of spin flips the elements of $\mathbb{C}[Z^H]$ must be invariant under these spin flips as well. 
If a subset of the vectors spanning $Z^{H}$ supports an invariant $I$ of the maximal Abelian subgroup a  spin flip $g^{\alpha}\in{G^{z}}$ maps this set of vectors to a new set supporting an invariant $I^{\alpha}$. If the number of qubits involved in the spin flip is even or if the bidegree $(a,b)$ of $I$ is such that $a-b$ is divisible by four, then the sum $I+I^{\alpha}$ is invariant under the spin flip. Let the index $\alpha$ run over all the possible sets of vectors related by spin flips including the identity operation.
Then the invariants of $N_G(H)$ are the sums $\sum_{\alpha}I^{\alpha}$ where either the bidegree $(a,b)$ of $I$ is such that $a-b$ is divisible by four or all spin flips involve an even number of qubits.
The $N_G(H)$ invariants are thus a subring of the invariants of the maximal Abelian subgroup of $G$ that contains $H$.

With the general form of $\mathbb{C}[Z^H]$ known we can state the following.

\begin{theorem}
\label{thew}
The closure of the set of states for which a group $H$ that is diagonalizable by local unitary operations with a non trivial action on each qubit is the symmetry group is an irreducible variety.
\end{theorem}
\begin{proof}
If $X^{H}$ is not irreducible it must be possible to express it as the union of sets where one or more of the polynomials in $\mathbb{C}[X^{H}]^{N_G(H)}$ vanishes. However for a generic element of $X^{H}$ which is a linear combination of all basis vectors in $S(H)$ all polynomials are non-vanishing. Hence, $X^{H}$ is irreducible.

\end{proof}
Thus, for each locally diagonalizable symmetry group $H$ that does not act trivially on any qubit the function field of $N_{G}(H)$ invariants on $X^{H}$ is isomorphic to the function field of $\mathrm{SU(2)}^{\times{n}}$ invariants on the symmetry stratum $X^{(H)}$. Furthermore, if $X^{(H)}$ is normal there is an isomorphism between $\mathbb{C}[Z^{H}]$ and $\mathbb{C}[Z^{(H)}]$ by Theorem \ref{luna}, i.e., in this case the structure of the ring of polynomial entanglement invariants is directly determined by the symmetry. 
We now make a few comments on how the individual invariants of $N_G(H)$ relate to different types of $G$ invariants and how the vanishing of the different invariants relate to different types of substrata of larger symmetry groups. 

As described in Refs. \cite{OS09,johanssonosterloh13} the irreducibly balanced monomials of bidegree $(a,b)$ where $a$ or $b$ is zero are related to ${\mathrm{SL(2)}}^{\times{n}}$ invariants. 
On a set $\mathcal{S}_{i}(H)$ of states in $Z^H$ corresponding to an irreducibly balanced set $s_i(H)$ any restriction of an ${\mathrm{SL(2)}}^{\times{n}}$ invariant polynomial to $Z^H$ is a power of $m_i$ \cite{OS09,johanssonosterloh13}. 
Thus, the irreducibly balanced monomials of this type are associated with the locally diagonalizable symmetry groups of SLOCC semistable states. 
When all the irreducibly balanced monomials of bidegree $(a,b)$ where $a$ or $b$ is zero vanishes, this indicates that the state is SLOCC unstable. If there are additional irreducibly balanced monomials this means that the symmetry stratum includes both SLOCC semistable and SLOCC unstable states. If all the irreducibly balanced monomials vanish this indicates a substratum where the symmetry group contains an Abelian subgroup which in turn includes $H$ as a proper subgroup. This substratum is by necessity a subset of the SLOCC unstable states.

A special case is when $S(H)$ of a locally diagonalizable a group contains a single irreducibly balanced set of vectors corresponding to a single pair of irreducibly balanced monomials related by complex conjugation of bidegree $(d,0)$ and $(0,d)$ respectively.

\begin{theorem}
\label{rty}
Assume that the generators of $\mathbb{C}[Z^{H}]$ include only one irreducibly balanced monomial $m_i$ up to complex conjugation. Assume further that the bidegree of this monomial is $(d,0)$.
Then there is a a single generator, up to complex conjugation, of the ${\mathrm{SL(2)}}^{\times{n}}$ invariant polynomials in $\mathbb{C}[Z^{(H)}]$ and it is of bidegree $(rd,0)$ for some $r$.

\end{theorem}

\begin{proof}

Assume that $P_{1}$ and $P_{2}$ are two polynomials of $\mathbb{C}[Z^{(H)}]$ that does not contain complex conjugated coefficients of the state vector. Their bidegrees must be multiples of $(d,0)$. Assume the bidegree of $P_1$ is $(r_1d,0)$ and the bidegree of $P_2$ is $(r_2d,0)$.

By B{\'e}zout's identity there is an element $P$ of the function field $\mathbb{C}(Z^{(H)})$ of bidegree
$(gcd(r_1,r_2)d,0)$, where $gcd(r_1,r_2)$ is the greatest common divisor of $r_1$ and $r_2$.

Consider $P^{\frac{r_1}{gcd(r_1,r_2)}}$ and $P_1$. These two invariants are of the same bidegree. Their restrictions to $Z^H$ are up to a complex factor both equal to $m_i^{r_1}$. Compensating for this complex factor the two invariants have the same value everywhere on $Z^H$. Since every state in $Z^{(H)}$ is local unitary equivalent to a state in $Z^H$ the two invariants have the same value everywhere on $Z^{(H)}$. Therefore, their difference is not in $\mathbb{C}(Z^{(H)})$ and they are equivalent as elements of $\mathbb{C}(Z^{(H)})$, i.e.,

 \begin{eqnarray}P_1=P^{\frac{r_1}{{\mathrm{gcd}}(r_1,r_2)}}.\end{eqnarray}
 Thus, $P_1$ is a power of $P$.
By repeating the argument it can be seen that there is an element in
$\mathbb{C}(Z^{(H)})$ such that all ${\mathrm{SL(2)}}^{\times{n}}$ invariants in $\mathbb{C}[Z^{(H)}]$ are powers of it. This element is not a quotient of polynomials since this would imply that all ${\mathrm{SL(2)}}^{\times{n}}$ invariants of $\mathbb{C}(Z^{(H)})$ were quotients. Thus, $\mathbb{C}[Z^{(H)}]$ has only one generator of the ${\mathrm{SL(2)}}^{\times{n}}$ invariant polynomials.

\end{proof}
Theorem \ref{rty} shows that for the special case of SLOCC semistable states with a locally diagonalizable symmetry group that acts non-trivially on all qubits, and for which there is a single irreducibly balanced set $s_{i}(H)$ of vectors defining a canonical form, there is only one ${\mathrm{SL(2)}}^{\times{n}}$ invariant polynomial among the generators of $\mathbb{C}[Z^{(H)}]$. 
In other words such a symmetry can be directly associated with a unique entanglement measure derived from the polynomial. However, it must be stressed that this does not exclude the possibility that several locally diagonalizable symmetries may be associated with the same ${\mathrm{SL(2)}}^{\times{n}}$ invariant polynomial in which case their symmetry strata are distinguished by other invariants.

The locally diagonalizable symmetry groups of SLOCC unstable states falls in two categories, those which are subgroups of ${D\times\mathrm{SU(2)}}^{\times{n}}$ for some discrete cyclic group $D$, and those which are subgroups only of ${\mathrm{U(1)}}\times{\mathrm{SU(2)}}^{\times{n}}$.
The first type is described by the irreducible monomials with bidegree $(a,b)$ where $0\neq a\neq b\neq 0$ and $|a-b|=|D|$. 
The relation between this type of monomials and ${\mathrm{SU(2)}}^{\times{n}}$ invariant polynomials is described in \cite{johanssonosterloh13}. 
The locally diagonalizable symmetry groups which are subgroups only of ${U(1)\times\mathrm{SU(2)}}^{\times{n}}$ have symmetry strata whose coordinate rings are polynomials made up from factors of the type $|c_{j}|^{2}$.
These are related to ${\mathrm{U(2)}}^{\times{n}}$ invariant polynomials on $Z^{(H)}$ of bidegree type $(a,a)$ since any such polynomial is made up of factors of the type $|c_{j}|^{2}$.

We end by commenting on the role of invariants that are made up of of the factors $|c_{j}|^{2}$
in the coordinate ring of a stratum of a locally diagonalizable symmetry which acts non-trivially on all qubits. 
These are related to substrata corresponding to non-Abelian symmetries and to Abelian symmetry groups that are not locally diagonalizable. From Theorem \ref{kullo} we have that the existence of a symmetry of this type required the state to satisfy a condition on the coefficients of the basis vectors in $S(H)$ equivalent to the at least one reduced one-qubit density matrix being maximally mixed. These conditions are of the form 

\begin{equation}\label{rr}\sum_{j\in{S_{k0}}}|c_j|^2-\sum_{j\in{S_{k1}}}|c_j|^2=0,\end{equation}
where $S_{k0}$ and $S_{k0}$ are the sets of coefficients of basis vectors in $S(H)$ with the $k$th qubit being in state $0$ and $1$ respectively. If $N_G(H)$ does not contain any spin flips the polynomial in the left hand side of Eq. \ref{rr} is in $\mathbb{C}[Z^H]$. If $N_G(H)$ contains spin flips there will be a polynomial in $\mathbb{C}[Z^H]$ containing the left hand side of Eq. \ref{rr} as a factor. 
Thus, a substratum corresponding to a non-Abelian or locally non-diagonalizable Abelian symmetry is always the zero locus of one or more polynomials of this type in $\mathbb{C}[Z^H]$. 

Theorem \ref{kullo} thus implies that a symmetry group with non-Abelian or locally non-diagonalizable Abelian action on the state space of the $k$th qubit is possible only if the state of this qubit contains no local information.
A symmetry group with this type of action on every single qubit state space can thus occur within $Z^{(H)}$ only if all reduced density matrices are maximally mixed, i.e., only for the maximally entangled states. Examples of non-Abelian symmetry groups of this type are those of the $GHZ$ states \cite{carteret} and cluster states \cite{briegelraussdorff}. In the symmetry stratification we thus find these maximally entangled states among the smallest strata.
Moreover, for the permutation invariant states Theorem \ref{kullo} together with the permutation invariance implies that non-Abelian symmetries or locally non-diagonalizable symmetries inside stratum closures of locally diagonalizable symmetry groups with non-trivial action on all qubits occur only for the maximally entangled states.

\section{Discussion}
\label{sec6}

We have reviewed the symmetry stratifications of Hilbert space and seen that it is possible to choose a set of entanglement invariants of an $n$-qubit system such that the symmetry strata of the system correspond to varieties defined by the zero locus of one or more of the invariants. There is thus always a set of functions of entanglement that go to zero when the system transitions to a more symmetric state.
Moreover, the same set of functions describe the transition between two given symmetry strata regardless of where in the stratum boundary the transition occurs.
Local unitary symmetry is an important manifestation of entanglement and are relevant for many quantum information tasks and for the understanding of different phases of matter. Therefore, a description of entanglement that naturally captures these symmetries may be useful.

The problem of describing entanglement on a symmetry stratum of a group $H$ can be reduced to the problem of describing the invariants under the group which preserves the symmetry $H$. 
In the presence of a symmetry a description of this kind may simplify the analysis of entanglement properties.
Furthermore, we discussed conditions for when the ring of polynomial entanglement invariants on the symmetry stratum or alternatively its function field is isomorphic to the ring of invariants of the symmetry preserving operations.

As a special case we studied the locally diagonalizable symmetry groups that have a non trivial action on each qubit. 
The closure of the set of states with a particular symmetry group $H$ of this type can be expanded in a set $S(H)$ of basis vectors. This set of vectors thus serve as a canonical form for all the states in the closure of the symmetry stratum. Using the canonical form one can construct the coordinate ring $\mathbb{C}[Z^{H}]$ which is sufficient to describes the entanglement properties of states within the symmetry stratum. In other words, $\mathbb{C}(Z^{H})$ distinguishes inequivalently entangled states once they are brought to the canonical form. We also described conditions for when a SLOCC invariant can be directly associated to a stratum of such a symmetry.

We end by commenting on the role of symmetries in some recently proposed classification schemes for entanglement.
Several classifications of multipartite entanglement using invariants and canonical forms with the aim of achieving an arrangement of entangled states into a finite number of classes have been proposed. In particular, we consider Refs. \cite{verstraete, viehmann, holweck, johanssonosterloh13} where the case of four qubits has been studied. In the first of these, Ref. \cite{verstraete}, the infinite set of SLOCC entanglement classes is arranged into nine different families. While thus achieving a tractable classification scheme this arrangement does not distinguish between states with qualitatively different entanglement properties such as for example the four-qubit GHZ state and the cluster states. 

These qualitatively different types of entanglement can be distinguished by the polynomial invariants.
For four qubits the subring of polynomial invariants that are $\mathrm{SL(2)}^{\times{4}}$ invariant is generated by four polynomials \cite{luq}.
In Ref. \cite{OS05} maximally entangled states representing inequivalent types of four-qubit entanglement were found using $\mathrm{SL(2)}^{\times{4}}$ invariants constructed through the invariant-comb approach. 
These states are the GHZ-state $\frac{1}{\sqrt{2}}(|1111\rangle+|0000\rangle)$, the so called X-state $\frac{1}{\sqrt{6}}(\sqrt{2}|1111\rangle+|1000\rangle+|0100\rangle+|0010\rangle+|0001\rangle)$, and the two states $\frac{1}{2}(|1111\rangle+|1100\rangle+|0010\rangle+|0001\rangle)$ and $\frac{1}{2}(|1111\rangle+|1010\rangle+|0100\rangle+|0001\rangle)$, related by a permutation of qubits, which are local unitary equivalent to cluster states. 

A classification scheme for four qubits that distinguishes between these qualitatively different types of entanglement was proposed in Ref. \cite{viehmann}. In this scheme the four generators of the ring of
$\mathrm{SL(2)}^{\times{4}}$ invariants are chosen such that each of them is non-vanishing only on one of the four states in Ref. \cite{OS05}. A similar classification scheme based on the subring of invariants with bidegrees $(a,b)$ for $a\neq{b}$ has been discussed in \cite{johanssonosterloh13}.

The classification schemes in Refs. \cite{viehmann, johanssonosterloh13} can be understood in terms of symmetry strata in the following way.
Each of the four states representing different entanglement types in \cite{OS05} is such that the basis vectors define a canonical form of the symmetry stratum of a non-trivial locally diagonalizable symmetry group. 
For each of these symmetry groups the canonical form is a single irreducibly balanced set of vectors. 
Therefore, each of the four invariants in \cite{viehmann} is the unique generator of the $\mathrm{SL(2)}^{\times{4}}$ invariants in its respective symmetry stratum. 
Therefore, this classification automatically captures the structure of the four symmetry stratum closures associated with the four states. 
The classification in \cite{johanssonosterloh13} works in the same way but distinguishes a larger set of symmetry strata since it uses a larger set of invariants.

Finally, we comment on the geometric classification scheme for four qubits presented in Ref. \cite{holweck}. The purpose of this scheme is to create a geometric picture of the different kinds of entanglement in terms of algebraic varieties combining the approaches of using invariants and canonical forms. This scheme does not explicitly take symmetries into account and the description in \cite{holweck} does not include the full four-qubit state space, but it still distinguishes between some of the symmetry strata including those of the four-qubit W and GHZ states as well as that of the X state.

In conclusion we can see that the classification of entangled states in terms of symmetry strata is  closely related to some of the already existing classifications.

\section{Acknowledgements}The author is grateful to Ingrid Irmer for help as well as for comments and discussions. The author also thank Andreas Osterloh, Jens Siewert, and Robert Zeier for discussions and Markus Grassl for discussions and useful comments, as well as for pointing out a mistake in a previous version.
Support from the National Research Foundation and the Ministry of
Education (Singapore), the Marie Curie COFUND action through the ICFOnest program, and the John Templeton Foundation is acknowledged.

\end{document}